\documentclass[a4paper,12pt]{article}
\usepackage[english]{babel}

\usepackage{amssymb,amsthm,amsmath}
\usepackage{graphicx}
\usepackage{enumerate}
\usepackage{array}
\usepackage{xcolor}
\usepackage{hyperref}

\newcolumntype{C}[1]{>{\centering\let\newline\\\arraybackslash\hspace{0pt}}m{#1}}
\newcolumntype{P}[1]{>{\centering\arraybackslash}p{#1}}

\newtheorem{theorem}{Theorem}

\newtheorem{lemma}{Lemma}

\newtheorem{corollary}{Corollary}

\newcommand{\Z}{\mathbb Z}
\newcommand{\N}{\mathbb N}

\newcommand{\R}{\mathbb R}

\newcommand{\ul}[1]{\underline{#1}}

\newcommand{\pre}[1]{\overset{\leftarrow}{#1}}

\newcommand{\prey}[1]{h^{-}({#1})}

\newcommand{\EXTRA}[1]{}

\title{A recursive function coding number-theoretic functions}

\author{Vesa Halava\thanks{Supported by emmy.network foundation
under the aegis of the Fondation de Luxembourg.}, Tero Harju, Teemu Pirttimäki
\\ Department of Mathematics and Statistics\\ University of Turku, Finland\\
Email: \texttt{\{vesa.halava, tero.harju, tealpi\}@utu.fi}
}
\begin{document}

\maketitle

\begin{abstract}
We show that there exists a fixed recursive function $e$ such that 
for all functions $h\colon \N\to \N$, 
there exists an injective function $c_h\colon \N\to \N$ such that $c_h(h(n))=e(c_h(n))$ for all $n \in \N$;
i.e., $h=c_h^{-1}ec_h$.
\end{abstract}

\section{Introduction}

This article is motivated by Woodin's theorem on computability~\cite{Woodin} and its entertaining implication stating that ``every function is computable'', where functions are considered to be on natural numbers, and ``computable'' refers to Turing computability 
in the models of Peano arithmetic. 
To a computer scientist or a mathematician coming from formal languages 
this seems quite strange. Non-computable functions should not be computable! 

The actual Woodin theorem on computability considers computations of Turing Machines in nonstandard models of arithmetic. Indeed, the Woodin theorem states a much stronger  model-theoretical result of computability in non-standard models and the existence of  end-extension models. We do not state the  Woodin theorem in full, because it is quite involved and lies outside the scope of 
our treatment. We shall shortly return to Woodin's theorem after we have introduced our main result.

We shall show that there exists a ``universal'' recursive function $e \colon \N \to \N$ that ``computes'' all functions $h\colon \N\to \N$ after we encode the inputs and outputs by an injective coding $c_h$; i.e.,
$$
h(n)=c_h^{-1}e c_h(n) \text{ for all } n \in \N\,.
$$

The universality questions for \emph{real} functions go back to 
Sierpi\'nski~\cite{Sierpinski}; see Larson et al.~\cite{Larson}.
Rado~\cite{Rado} considered the problem in a more general setting.
Sierpi\'nski showed that 
under the assumption of the Continuum Hypothesis there exists
a fixed Borel function $B \colon \R^2 \to \R^2$ such that 
for every function $h \colon \R^2 \to \R$, there exists a function $\varphi$
satisfying $h(x,y)= B(\varphi(x),\varphi(y))$. Larson et al.~\cite{Larson}
show that without the Continuum Hypothesis, there may not exist
universal functions on the reals.

The idea behind Woodin's theorem is quite different from ours. Indeed, the Turing Machine computing all functions in Woodin's result is non-halting, implying that it does not compute any recursive or partially recursive function in the standard computational sense. The idea of the construction is to use nonstandard computation where the machine, given input $n$, outputs $h(n)$ after infinitely many steps. Nonstandard computation is a model-theoretical concept and the infinitely many steps corresponds to the nonstandard (infinite) elements in a nonstandard model of Peano arithmetic. These nonstandard elements are greater than any natural number. In ouir results the function is recursive in the standard computational setting.

The model theoretical property of end-extensions in the Woodin's theorem is (in a simple setting to present the idea) the following: Given an initial fragment or part of a function $h \colon \N \to \N$ and a model $\mathcal{N}$ (with universe $N$) of Peano arithmetics where the Turing machine computes the initial part, there exists a model $\mathcal{M}$ (with universe $M$) where the Turing machine computes the full function $h$, such that $\mathcal{M}$ is an \emph{end-extension} of $\mathcal{N}$ in the sense that $N \subseteq M$ and for all $m \in M$, if there exists $n \in N$ such that $m \leq n$, then $m \in N$.

In the treatment below, we consider only number-theoretic functions 
$h\colon \N \to \N$ of a single argument. Our approach does not
involve model theory like Woodin's theorem does.

\section{Trees of functions}

In our proofs we use directed graphs induced by functions on the set $\N = \{0,1,2,3, \dotsc\}$.  
Let $h\colon \N\to \N$ be a function and $B\subseteq \N$. 
The image of $B$ under $h$ is the set
$$
h(B)=\{h(n) \mid n\in B\}.
$$

We consider the directed (possibly infinite) graph $G_{(h,B)} = (V_{(h,B)},E_{(h,B)})$ allowing loops where the set of \emph{vertices} is  $V_{(h,B)}=B\cup h(B)$, and the set of \emph{edges} is $E_{(h,B)} = \{(n,h(n)) \mid n\in B \}\subseteq V_{(h,B)} \times V_{(h,B)}$. The edges are directed.
Such graphs induced by a function $h$ on some subset $B\subseteq \N$ are called \emph{function graphs}. 

Denote by $G_h=G_{(h,\N)}=(\N,H)$ the \emph{full (function) graph of} $h$, where $H=\{(n,h(n))\mid n\in \N\}$.

\begin{lemma}
A graph $G = (V,E)$, where $V \subseteq \N$, is a function graph if and only if all vertices have at most one outgoing edge.
\end{lemma}
\begin{proof}
Let $G_{(h, B)}$ be a function graph. Then, because $h$ is a function, for the vertices $n \in B$ there is exactly one value $h(n)$ and hence exactly one outgoing edge. The vertices in $h(B) \setminus B$ have no outgoing edges. Hence every vertex of $V_{(h,b)}$ has at most one outgoing edge.

Then, let $G = (V,E)$, where $V \subseteq \N$, be a graph such that every vertex has at most one outgoing edge. Now $E$ is a function on some subset of $V \subseteq \N$.
\end{proof}

Let $G=(V,E)$ be the graph of a function $h$. A vertex $h(n)$ is  
the \emph{successor of} $n\in V$ in $G$ for $h$. 
We let
\begin{align*}
h^+(n) &= \{v \mid v=h^k(n) \text{ for some } k \ge 1\},\\
\prey{n} &= \{v \mid n=h^k(v) \text{ for some } k \ge 1\}
\end{align*}
denote the sets of the \emph{descendants} and \emph{predecessors} of $n$,
respectively. 

Note that the set of \emph{immediate predecessors} $h^{-1}(n)$ can be infinite.
It also can be empty.

A \emph{path} in a function graph is a sequence of vertices that is of the form $n, h(n), h^2(n),\dots ,h^k(n)$ for some $k \ge 1$ and no vertices repeat.
A \emph{cycle} in a function graph is a \emph{closed path}, that is, it is a sequence of vertices that is of the form $h(n), h^2(n),\dots ,h^k(n)=n$ 
for some $k\ge 1$, where the vertices are all different.
The number $k$ is called the \emph{length} of the cycle.   
In the case $k=1$, the sequence is simply $n$; this denotes a \emph{loop}, 
an edge of the form $(n,n)$.

For vertices $u$ and $v$ of a function graph, 
we say that $v$ \emph{reaches} $u$ or $u$ \emph{is reachable from} $v$ if there is a path from $v$ to $u$;
i.e., if $u$ is a descendant of $v$.

\begin{lemma}
Consider a graph of a function $h$.
If $u$ and $v$ belong to a common finite or infinite path, then
one of them is a descendant of the other:
either $h^k(u)=v$ or $h^k(v)=u$ for some $k\in  \N$.
\end{lemma}

\begin{proof}
This is clear from the definition of a function graph.
\end{proof} 

Let $n \in \N$. The \emph{connected component of $n$} in the full function graph $G_h=(\N,H)$ is the subgraph $G_h^n=G_{(h,V)}$ such that $n\in V$ and for all $m\in \N$, $m\in V$ if and if only if $n$ and $m$ have a common descendant. A subgraph of $G_h$ is called a \emph{connected component of $G_h$} if it is the connected component for some vertex $n \in \N$.

\begin{lemma}
A connected component of $G_h$ can have at most one cycle.
\end{lemma}

\begin{proof}
Follows from the fact that $h$ is a function. 
Indeed, if a connected component has a cycle, it is unique as all vertices have a unique outgoing edge. This means that the connected component of $G_h$ ``ends'' in the cycle.
\end{proof}

All vertices $n$ in the full function graph $G_h$ start 
an infinite sequence $(h^k(n))_{k=0}^\infty$. 
This sequence may be ultimately periodic, 
which means that the iteration of $h$ starting from $n$ leads to a cycle.  
Otherwise, the sequence is an infinite path.

\begin{lemma}\label{form1}
The connected component $G_h^n$ of $G_h$ either 
\begin{enumerate}
\item has a unique cycle and for all vertices $m$ in $G_h^n$, the sequence 
$(h^k(m))_{k=0}^\infty$ is ultimately periodic and leads to the unique cycle, 
\item has no cycles and for all vertices $m$ in $G_h^n$, the sequence 
$(h^k(m))_{k=0}^\infty$ is not ultimately periodic.
\end{enumerate}
For all vertices $m$ in the component $G_h^n$, there exist natural numbers $k_1$ and $k_2$ such that $$h^{k_1+i}(n)= h^{k_2+i}(m)$$
for all $i\in \N$. 
\end{lemma}

\begin{proof}
This follows from the fact that all vertices of the $G_h^n$ have 
a common descendant with $n$. Hence, for all vertices $m$, 
$(h^k(m))_{k=0}^\infty$ is ultimately periodic iff   $(h^k(n))_{k=0}^\infty$ ultimately periodic. 

Also if $m$ and $n$ have a common descendant, they both reach a vertex $w=h^{k_1}(n)=h^{k_2}(m)$. 
\end{proof}

In order to characterize the connected components of a function graph
we need function trees.
A subgraph $G=(V,E)$ of a function graph is a \emph{function tree} if there is a unique \emph{root} $u\in V$ such that all vertices in $V \setminus \{u\}$ reach $u$. Naturally, all function trees are trees, and they can be infinite. The \emph{depth} of a function tree is the length of the longest path in it.
Clearly all finite trees have finite depth.  

Let $n\in \N$ be a vertex of $G_h$ for a function $h$ and let 
$C_h^n$ be the set of the vertices in the cycle of $G_h$ containing $n$.
Hence, if $n$ does not belong to a cycle then $C_h^n=\emptyset$. 
Denote by 
$$
\pre{n}=\{v \mid \exists k\ge 1 \colon h^k(v)=n \text{ and } h^{k-1}(v)\notin C_h^n \},
$$ 
the set of all predecessors of $n$ that reach $n$ via a path not visiting 
the cycle of $n$ (if $n$ belongs to one), that is, reach $n$ in the subgraph of $G_{(h,B)}$ induced by $(V_{(h,B)} \setminus C_n^h) \cup \{n\}$. If $n$ has 
only one immediate predecessor and it is in $C_h^n$, then $\pre{n}=\emptyset$. If the vertex $n$ does not belong to a cycle, then $\pre{n}=h^{-}(n)$.
Let then 
$$
O_h^n=G_{(h,\pre{n})}.
$$

\begin{lemma}
Subgraph $O_h^n$ is a function tree with $n$ as a root.  
\end{lemma}

\begin{proof}
By the definition of the sets $\pre{n}$, we have $O_h^n=G_{(h,\pre{n})}=\left(V_{(h, \pre n)},E_{(h, \pre n)}\right)$ where $V_{(h, \pre n)} = \pre{n}\cup h(\pre{n})=\pre{n}\cup\{n\}$, since $n$ is the only element of $h(\pre{n})$ not in $\pre{n}$. By the construction of $O_h^n$, the vertex $n$ is reachable from any other vertex and there are no cycles. There is also no edge from $n$ to any other vertex, making $O_h^n$ is a function tree with $n$ as its unique root.
\end{proof}

\begin{lemma}\label{lem:formofcc}
In the full function graph $G_h=G_{(h,\N)}$,
the connected components $G_h^n$, where $n \in \N$, are either 
\begin{enumerate}
\item a cycle $C_h^m$ together with the all function trees $O_h^{u}$ with $u\in C_h^m$, or  
\item the function tree $O_h^n$ together with the function trees $O_h^m$ 
for the (infinitely many) descendants $m$ of $n$. 
\end{enumerate}  
\end{lemma}

\begin{proof}
By Lemma~\ref{form1} we have two cases:

1) If  a connected component contains a cycle, then it is of the first form and $u = h^k(n)$, where $k$ is the smallest natural number for which $h^k(n) \in C_h^m$.

2) On the other hand, if the connected component does not contain a cycle, 
then it is a directed tree such that from all of its vertices 
there begins an infinite path. 
By definition, the root of $O_h^n$ is $n$. 
The other vertices in the component have a common descendant with $n$, that is, 
they belong to a function tree $O_h^m$ for some $m\in h^+(n)$.
This proves the claim.
\end{proof}

\section{Universal recursive function $e$}

There are countably many non-isomorphic connected full function graphs and these
are partitioned into two different classes according to whether they contain a cycle.
Moreover, each class contains infinitely many non-isomorphic graphs, but we set out to show that they can all be embedded
in a graph including infinitely many cycles (of all possible lengths) and an infinite number of function trees as subgraphs. The basic idea in the construction of the special function $e$ is 
to pack its function graph with infinitely many times each of the 
above possibilities.

The construction of $e$ can be best explained by considering its image sequence   
$\sigma =e(0), e(1), e(2), \dotsc$. In the sequence $\sigma$, all natural numbers 
occur infinitely many times. It can be separated to two cases: 
(1) First, for the powers $n=2^k$, the iterative sequence $(e^i(n))_{i=0} ^\infty$ 
is defined so that it is not ultimately periodic for all $k \in \N$. 
(2) Secondly, all natural numbers are roots of a function tree with infinite depth. 
For any ultimately periodic sequence of natural numbers, 
the full function graph of $e$ contains infinitely many cycles of all lengths. 
  
We shall now give a construction of the sequence $\sigma$. For a variable $x$ and a function $f$, we write $x \leftarrow f(x)$ for substitution, that is, $x \leftarrow f(x)$ means ``evaluate $f(x)$ and replace the value of $x$ by the result''.

\begin{enumerate}
\item[(i)] For all numbers $x \in \N$ of the form $x=2^k$ for $k \ge 1$ (the nontrivial powers of $2$), let $e(x)=2^{2k}$. These values do not change in later steps.

\item[(ii)] We use three variables $k$, $n$ and $i$, where 
\begin{itemize}
\item
$k$ keeps track of the length of the largest cycle, 
\item
$n$ keeps track of the largest value of $x$ for which $e(x)$ is determined such that $x$ is not a power of $2$, and 
\item
$i$ keeps track of the current number in the sequence. 
\end{itemize}
Set the initial values $k \leftarrow 1$, $n \leftarrow 1$ and $i \leftarrow 0$. Then repeat steps (a) and (b) below in alternating order \emph{ad infinitum}:
\begin{enumerate}
\item Create cycles of all lengths $j = 1, \dotsc, k$. The creation of a cycle is as follows: In the case $j = 1$, we have a fixed point $e(i) = i$. Otherwise,
if there are no powers of $2$ in  $\{i,\dots,i+(j-1)\}$, 
then the cycle of length $k$ will be 
\begin{align*}
e(i) 			& = i + j - 1	\\
e(i + 1) 		& = i			\\
e(i + 2)		& = i + 1   	\\
				& \vdots		\\
e(i + j - 1)	& = i + j - 2	
\end{align*}

If the number of nontrivial powers of $2$ in 
$\{i,\dots, i+(j-1)\}$ is $t$, then let $e(i)=i+(j-1)+t$ and skip over the powers of $2$ from both the image and domain when defining the values. Actually, it is an easy number-theoretic exercise to show that in our construction, $t\le 1$.

After the creation of each individual cycle, set $i \leftarrow i + j$. After all the cycles have been created, set $k \leftarrow k+1$.

\item Set $n$ to the the largest value of $x$ for which $e(x)$ is determined such that $x$ is not a power of $2$, and then set $e(i) = 0, e(i + 1) = 1, \dotsc, e(i + n) = n$. Again we skip over the powers of $2$; in every step, $i$ is tested, and if $i=2^k$ for some $k \geq 1$, we set $i \leftarrow i+1$.
\end{enumerate}
 
\end{enumerate}

The first 56 elements of the sequence $\sigma$ are
\begin{equation*}
\begin{split}
&\stackrel{(0)}{\underline{0}},\stackrel{(1)}{0},\stackrel{(2)}{2^2},\stackrel{(3)}{\ul{3}},\stackrel{(4)}{2^{2\cdot 2}},\underline{\stackrel{(5)}{6},\stackrel{(6)}{{5}}},\stackrel{(7)}{0},\stackrel{(8)}{2^{2\cdot 3}},\stackrel{(9)}{1},\stackrel{(10)}{2},\stackrel{(11)}{3},\stackrel{(12)}{4},\stackrel{(13)}{5},\stackrel{(14)}{6},\ul{\stackrel{(15)}{15}}, \stackrel{(16)}{2^{2\cdot 4}}, \\
&\underline{\stackrel{(17)}{18},\stackrel{(18)}{17}},\underline{\stackrel{(19)}{21},\stackrel{(20)}{19},\stackrel{(21)}{20}},\stackrel{(22)}{0},\dots, \stackrel{(31)}{9}, \stackrel{(32)}{2^{2\cdot 5}}, \stackrel{(33)}{10},\dots, \stackrel{(44)}{21},\underline{\stackrel{(45)}{45}}, \\
&\underline{\stackrel{(46)}{47},\stackrel{(47)}{46}},\underline{\stackrel{(48)}{50},\stackrel{(49)}{48},\stackrel{(50)}{49}},\underline{\stackrel{(51)}{54},\stackrel{(52)}{51},\stackrel{(53)}{52},\stackrel{(54)}{53}}, \stackrel{(55)}{0}, \dots
\end{split}
\end{equation*}
where $(i)$ refers to the position in the sequence, and $e(i)$ is written below $(i)$. Underlining denotes the cycles. For example, $e(19)=21, e(20)=19, e(21)=20$ makes a cycle of length three in the function graph.

\begin{theorem}
The function $e$ is recursive.
\end{theorem}

\begin{proof}
This is clear from the above algorithm for the sequence $\sigma$. 
Indeed, for an input $m \in \N$, we can add a new counter to stop the computation when $e(m)$ is reached.
\end{proof}

The function $e$ has the following properties:

\begin{enumerate}

\item The image set $e(\N) = \N$.

\item For all natural numbers $n$, there are infinitely many $x \in \N$ such that $e(x) = n$. 



\item Let $P=\{2^k\mid k\ge 1\}$, $C=\{n\mid n \text{ occurs in a cycle in }G_e\}$, and $T=\N\setminus (P\cup C)$. These sets are all infinite and (pairwise) disjoint, and $\N=P\cup C \cup T$, that is, they form a partition of $\N$.

\item Let $T_n = \pre{n} \setminus P$ for all $n \in P \cup C$. These sets form a partition of $T$.

\item For all $n \in T$, we have $e(n)<n$.

\item For each $k \in \N$, there are infinitely many cycles of length $k$ in $G_e$.

\item For all $m \in \N$, the sequence $(e^i(2^m))_{i = 0}^\infty$ is not ultimately periodic (these correspond to infinite paths in $G_e$). Naturally, there are infinitely many disjoint sequences like this; there is one for each odd value of $m$.

\item For all $m \in \Z_+$, if $k$ is the largest odd factor of $m$, then there exists a unique $i \in \N$ such that $2^m = e^i(2^k)$. If $m$ is odd, then $i = 0$.


\item For all $n\in P\cup C$, the graph $G_{(e,T_n)}$ is a function tree with root $n$, and each vertex of $G_{(e,T_n)}$ has infinitely many incoming edges. 

\item For $n \in C$, $O_e^n=G_{(e,T_n)}$. This follows from the fact that $\pre{n}=T_n$.     

\item For $n \in P$ we have two cases: 
\begin{enumerate}
\item if $n=2^k$ with odd $k$, then 
$O_e^n=G_{(e,T_n)}$. This follows from the fact that $\pre{n}=T_n$.

\item if $n=2^k$ with even $k$, then $\pre{n}=T_n \cup \{2^{\frac k2}\}\cup \pre{2^{\frac k2}}$, and 
\[
O_e^n = \left(\pre{n} \cup \{n\}, E_{(e,T_n)}\cup E_{(e,U_n)}\right),
\]
where $U_n=\left\lbrace 2^{\frac k2}\right\rbrace\cup \pre{2^{\frac k2}}$.
Note that $G_{(e,U_n)}$ consists of $O_e^{2^{\frac k2}}$ 
together with the edge $(2^{\frac k2},n)$.

\end{enumerate} 

\item Combining properties 7 and 9(a), we get that $G_e$ has a countably infinite number of connected subgraphs that contain two-way infinite paths. Indeed, for all $n=2^k$ with odd $k$, $O_e^n$ is a function tree of infinite depth with root $n$, and, for all vertices $v$, the sequence $(e^k(v))_{k=0}^\infty$ is not ultimately periodic.      

\end{enumerate}

We define further, for all vertices $n\in P\cup C$, 
$$V_n = V_{(e, T_n)} = T_n\cup \{n\} \text{ and } E_n = E_{(e, T_n)} = \{(x,e(x)) \mid x \in T_n \}.$$


In the following, an \emph{enumeration (function)} $\varepsilon$ of a set $S$ is an injective mapping from $\N$ to $S$ such that $\varepsilon(n)$ is defined for all $0 \le n < |S|$. Hence, for a finite set $S$, the function $\varepsilon$ defines a one-to-one correspondence between the sets $\{0,1,\dots,|S|-1\}$ and $S$. Note that $S$ may be countably infinite.

Let $n\in P\cup C$. For each $v\in V_n$, we define an enumeration $\beta_{v}$ of the set $e^{-1}(v) \cap T_n$ based on the function $e$. We need to intersect $e^{-1}(v)$ with $T_n$: For $n \in C$, $e^{-1}(n)$ contains also the unique predecessor $m$ of $n$ in the cycle and $m \notin T_n$. Similarly in the case where $n=2^k\in P$ with even $k$, $2^{\frac{k}2}\in e^{-1}(n)\setminus T_n$ (see property 11(b) of $e$).
 
Denote $\sigma_j = e(j)$ for all $j \in \N$. For each $v \in V_n$, we define a sequence $(x_{v,i})_{i=0}^\infty$ as the subsequence of $\sigma$ that consists of exactly those $\sigma_j$ for which $e(\sigma_j) = v$ and $\sigma_j \geq v + 2$. Now let $\beta_v(i)= x_{v,i}$ for all $v \in V_n$ and all $i \in \N$. Then for all $v \in V_n$, $\beta_v$ is an enumeration of the set $e^{-1}(v)\cap T_n$.

The condition $x_{v,i} \geq v+2$ is necessary to exclude vertices in $P \cup C$ from the enumeration: If $n\in C$, then the unique predecessor of $n$ in the cycle of length $t$ is either $n+1$ or $n-(t-1)\le n$. On the other hand, if $n=2^k$ with even $k$, then $n$ has the predecessor $2^{\frac{k}2}$ which is not in $T_n$ and is clearly less than $n$. Hence
\[
e^{-1}(v) \cap T_n = \beta_v(\N) = \{n \in \N \mid e(n) = v \text{ and } n \geq v + 2\}.
\]

In the next lemma we show that every function tree is isomorphic to a subgraph of $G_{(e,T_n)}$ for some $n \in \N$. Instead of $G_{(e,T_n)}$, we could also show the result for every function tree 
with root $m\in T_n$ in $G_e$, but that is unnecessary for what follows.  

\begin{lemma}\label{lem:ot_tn}
Let $h\colon \N\to \N$ be a function and $O=(V,E)$ be a subgraph of $G_h$ that is a function tree with root $r$. 
Then, for all $n\in P\cup C$, there is an injective function $c_O\colon V \to V_n$
such that $c_O(h(m))=e(c_O(m))$ for all $m \in V\setminus\{r\}$, and $c_O(r)=n$.  
\end{lemma}

\begin{proof}
Let $\pi$ be an enumeration function of the vertex set $V$. We may assume that $\pi(0)=r$. Also, let $\alpha_u$ be an enumeration function of the set $h^{-1}(u)\cap V$ 
for each vertex $u\in V$.  

We define $c_O$ by an ascending sequence of subsets $\Delta_j \subseteq V \times V_n$. In every step we add pairs $(a,b)$ to $\Delta_j$ such that $a\in h^{-1}(u)$ and $b\in e^{-1}(v)$ for some $(u,v)\in \Delta_j$. Therefore, $h(a)=u$, $e(b)=v$ and $c_O(h(a))=c_O(u)=v=e(b)=e(c_O(a))$. We will eventually define $c_O$ in such a way that $\Delta_j\subseteq c_O$ for all $j \in \N$.
   
To begin with, let 
$$
\Delta_0=\{(r,n)\} \cup \{(\alpha_{r}(i), \beta_{n}(i))
\mid i =  0,\dots, |h^{-1}(r)\cap V| - 1 \},
$$ 
where $\beta_{n}$ is an enumeration of the set $e^{-1}(n)\cap T_n$. Now, for any pair $(\alpha_{r}(i), \beta_{n}(i))\in \Delta_0$, 
we have $h(\alpha_{r}(i))=r$ and $e(\beta_{n}(i))=n$, and $c_O(h(\alpha_{r}(i)))=c_O(r)=n=e(\beta_{n}(i))=e(c_O(\alpha_{r}(i)))$.

The sets $\Delta_j$ for $j \ge 1$ are defined inductively:
\begin{enumerate}
\item If $(\pi(j),v)\in \Delta_{j-1}$ for some $v\in V_n$, then we have two cases: 
\begin{enumerate}
\item If for all $u\in h^{-1}(\pi(j))$, there exists a vertex $v' \in V_n$ such that $(u,v')\in \Delta_{j-1}$, then 
set $\Delta_j=\Delta_{j-1}$. The construction works in such a way that if such a vertex $v'$ exists for \emph{some} predecessor $u$ then there is one for \emph{all} predecessors.

\item Otherwise, set
\[\Delta_j=\Delta_{j-1} \cup \{(\alpha_{\pi(j)}(i), \beta_{v}(i))
\mid i= 0,\dots, |h^{-1}(\pi(j))\cap V|-1\}.
\]
\end{enumerate}
\item If there is no $v\in V_n$ such that $(\pi(j),v)\in \Delta_{j-1}$, there exists a (finite and unique) shortest path $p_0,p_1,p_2,\dots,p_k$ in $O$, where $p_0 = \pi(j)$ and $(p_k,v)\in \Delta_{j-1}$ for some $v\in V_n$. Hence for some $k\ge 1$ and some $v\in V_n$, we must have $(h^k(\pi(j)),v)\in \Delta_{j-1}$ since $O$ is a function tree with root $r$ and at least $(r,n)\in \Delta_{j-1}$. Let $p_i=h^i(\pi(j))$ for $i=0,\dots, k$. Then let $v_k=v$ and construct sets $\Sigma_t$, counting downwards for $t=k,k-1,\dots, 0$, by defining
$$
\Sigma_t= \{(\alpha_{p_t}(i),\beta_{v_t}(i))\mid  i=0,\dots, |h^{-1}(p_t)\cap V| - 1 \},  
$$
and $v_{t-1}= v'$ such that $(p_{t-1},v')\in \Sigma_t$ for $t=k,\dots,1$. Let $\Sigma=\bigcup_{t=0}^k \Sigma_t$ and $\Delta_j=\Delta_{j-1} \cup \Sigma$. We note the following:
\begin{itemize}
\item For the path $p_0,p_1,p_2,\dots,p_k$, there are pairs $(p_i,v_i)\in \Sigma$ for $i=0,\dots, k$ such that $e(v_0)=v_1, e(v_1)=v_2,
\dots, e(v_{k-1})=v_k$, and $(p_k,v_k) \in \Delta_{j-1}$. Therefore, $h(p_i)=p_{i+1}$ and $e(v_i)=v_{i+1}$ for $i=0,\dots, k-1$, and $(p_i,v_i)\in \Delta_j\subseteq c_O$. Therefore, $c_O(h(p_i))=v_{i+1}=e(v_i)=e(c_O(p_{i}))$.  

\item For all pairs $(\alpha_{p_t}(i),\beta_{v_t}(i))\in \Sigma\setminus \{(p_t,v_t)\mid t=0,1,\dots, k\}$, we have 
$h(\alpha_{p_t}(i))=p_t$, $e(\beta_{v_t}(i))=v_t$, and $(p_t,v_t)\in \Sigma$. Therefore, 
$c_O(h(\alpha_{p_t}(i)))=c_O(p_t)=v_t=e(\beta_{v_t}(i))=e(c_O(\alpha_{p_t}(i)))$.
\end{itemize}
\end{enumerate}

Now $\Delta_j\subseteq \Delta_{j+1}$ for all $j \in \N$. Set $\Delta=\bigcup_{j}^{\infty} \Delta_j$. Finally, define $c_O\colon V \to V_n$ by
$c_O(u)=v$ if and only if $(u,v)\in \Delta$.

The injectivity of $c_O$ is clear from the construction of $\Delta$.
\end{proof}

We shall now prove that for every function $h \colon \N \to \N$, every connected subgraph $G$ of $G_h$ can be embedded into $G_e$, that is, there is a subgraph of $G_e$ that is isomorphic to $G$. 
For all $k \in \Z_+$, let
\begin{align*}
C_k=\{n\mid n \text{ is the smallest number in some cycle of length } k \text{ in } G_e\}.
\end{align*}
We prove the claim first for subgraphs that contain a cycle.

\begin{lemma}\label{lem:cyclesg}
Let $h\colon \N\to \N$ be a function and let $G=(V,E)$ be a connected subgraph of $G_h^n$ with a cycle of length $k$. Then for all numbers $i\in C_k$ with $G_e^i=(W,F)$, 
there exists an injective function $c_i\colon V\to W$ such that $c_i(h(j))=e(c_i(j))$ for all $j\in V$.
\end{lemma}

\begin{proof}
By case 1 of Lemma \ref{lem:formofcc}, $G_h^n$ consists of a cycle with function trees rooted at each vertex on the cycle. Assume that the cycle in $G_h^n$ is $C^{m_1}_h$ so that $h(m_1)=m_2,\dots, h(m_{k})=m_1$, and let
$O_h^{m_j}$ be the function trees rooted at $m_j$ for $j=1,\dots, k$. 

Let $i\in C_k$, and denote $C_e^{i}=\{t_1,\dots, t_k\}$ where $e(t_j)=t_{j+1}$ for $j=1,\dots, k-1$ and $e(t_k)=t_1$. We may assume $i=t_1$. 

We define $c_i$ again as a subset of $V\times W$. First
$$
I=\{(m_j,t_j)\mid j=1,\dots, k\}.  
$$ 
Then for each function tree $O_h^{m_j}$ in $G$,  let $c_{O_h^{m_j}}$ be the mapping given in  Lemma~\ref{lem:ot_tn} from $O_h^{m_j}$ to $G_{(e,T_{t_j})}$. Finally, let 
$$
c_i=I\cup \bigcup_{j=1}^k  c_{O_h^{m_j}}
$$

Now, by Lemma~\ref{lem:ot_tn}, we have $c_{O_h^{m_j}}(h(m)) = e(c_{O_h^{m_j}}(m))$ for all $m_j \in C_h^{m_1}$ and for all $m \in O_h^{m_j}\setminus \{m_j\}$. For all $m_j\in C^{m_1}_h$ and $j=1,\dots, k-1$, we have $c_i(h(m_j))=c_i(m_{j+1})=t_{j+1}=e(t_j)=e(c_i(m_j))$, and for $m_k$, we have  
$c_i(h(m_k))=c_i(m_{1})=t_{1}=e(t_k)=e(c_i(m_k))$. This proves the claim.

\end{proof}

\begin{lemma}\label{lem:infdessg}
Let $h\colon \N\to \N$ be a function and $G_h^n=(V,E)$ such that it does not contain cycles. Let $i=2^m$, where $m \in \N$ is odd, and denote $G_e^i=(W,F)$; then
there exists an injective function $c_i\colon V\to W$ such that $c_i(h(j))=e(c_i(j))$ for all $j\in V$.
\end{lemma}

\begin{proof}
Now, the sequence $(h^k(n))_{k=0}^\infty$ is not ultimately periodic. We define the mapping $c_i$ again as a subset of $V\times W$, by defining sets $I_k$ such that $c_i=\bigcup_{k=0}^\infty I_k$.

First, let
$$
I_0=\{(h^k(n),2^{2^k  m})\mid k \in \N \}\cup c_{O_n^h}
$$ 
where  $c_{O_n^h}$ is the mapping of Lemma~\ref{lem:ot_tn} from $O_n^h$ into $G_{(e,T_i)}$. In other words, $I_0$ defines $c_i$ for all predecessors of $n$ (i.e., for all vertices in $h^{-}(n)$), for the vertex $n$, and for all descendants of $n$.   
Indeed, $I_0$ maps the predecessors and descendants of $n$ correctly by Lemma~\ref{lem:ot_tn} and the fact that  $I_0(h^k(n))=2^{2^k m}=e(2^{2^{k-1} m})= e(I_0(h^{k-1}(n))$ for all $k \in \Z_+$. 

What remains to consider is the function trees $O_h^u$, where $u = h^k(n)$ for $k \in \Z_+$. 
Let 
$$
Y_u=h^{-}(u)\setminus \left(\{h^{k-1}(n)\}\cup h^{-}(h^{k-1}(n))\right).  
$$  
Hence $Y_u$ consists of all predecessors of $u = h^k(n)$ except the vertex
$h^{k-1}(n)$ and its predecessors.

We define $I_j$, for all $j \in \Z_+$, so that $\bigcup_{j=0}^{k-1}I_{j}$ already defines images for vertices in $\{h^{k-1}(n)\} \cup h^{-}(h^{k-1}(n)) =h^{-}(u)\setminus Y_u$, and $I_k$ defines the images for $Y_u$.

Let $I_k$ to be the mapping $c_{G_{(h,Y_u)}}$ from function tree $G_{(h,Y_u)}$ into $G_{(e,T_{2^{2^k m}})}$, given by Lemma~\ref{lem:ot_tn}. Let $c_i=\bigcup_{k=0}^\infty I_k$. Since for all vertices $v\in V$, either $v \in O_h^n$, or $v \in Y_u$ for some $u = h^k(n)$ and $k \in \Z_+$, or $v=h^k(n)$ for some $k \in \Z_+$, and since the claim holds for vertices in $O_n^h$ and in the sets $Y_u$ by Lemma~\ref{lem:ot_tn}, $c_i$ satisfies the claim. 
\end{proof}

We are ready to prove our main theorem.

\begin{theorem}\label{univfunction}
Let $h\colon \N\to \N$ be a function. There exists an injective function $c_h\colon \N\to \N$ such that $c_h(h(n))=e(c_h(n))$ for all $n \in \N$.
\end{theorem}

\begin{proof}
Let $G_h=(\N,E)$ be the graph of $h$. In $G_h$, there are countably many connected subgraphs. Let $\pi$ be an enumeration function of these subgraphs. Now, for all $i \in \Z_+$, if $\pi(i)$ is 
\begin{enumerate}
\item a subgraph of type 1 in Lemma~\ref{lem:formofcc} with a cycle of length $k$, let $I_i$ be the mapping $c_t$ given in Lemma~\ref{lem:cyclesg} to the $i$th cycle of length $k$ (in other words, $t$ is the $i$'th number in $C_k$).

\item a subgraph of type 2 in Lemma~\ref{lem:formofcc}, let $I_i$ be the mapping $c_t$ 
given in Lemma~\ref{lem:infdessg}, where $t=2^{m}$ and $m$ is the $i$th odd number (so $m=2i-1$). 
\end{enumerate}   

Set $c_h=\bigcup_{i=1}^\infty I_k$. The claim now follows from Lemmas~\ref{lem:cyclesg} and \ref{lem:infdessg}. 
    
\end{proof}

Since the proof did not depend on arithmetical properties of $\N$ (such as summation or ordering) we can present a stronger version of the theorem:

\begin{corollary}
Let $Q$ be a countable set and let $h\colon Q\to Q$ be a function. There exists an injective function $c_h\colon Q\to \N$ such that $c_h(h(x))=e(c_h(x))$ for all $x \in Q$.
\end{corollary}

In addition to the above ``function formulation'', there are other viewpoints that give different interpretations for the result.

\begin{corollary}[The graph formulation]
Every function graph is isomorphic to a subgraph of $G_e$.
\end{corollary}

It is well known that every directed graph $G = (V,E)$ induces a preorder $\preceq$ on $V$ by the condition $u \preceq v$ if and only if $u$ reaches $v$ or $u = v$.

\begin{corollary}[The preorder formulation] \label{preord}
Every preorder induced by a function graph is order-isomorphic to the preorder induced by some subgraph of $G_e$.
\end{corollary}

Finally, we note that our result cannot have a form of the end-extension property mentioned in the introduction. In our version, the end-extension would correspond to embedding a finite enumeration function in $G_e$ and then embedding that graph to a new graph. Given a finite initial part of a function $h \colon \N \to \N$ (that is, a finite enumeration), it is easy to embed this initial part in $G_e$. But in order to choose an extendable initial coding, we need details of the connected components of the vertices of the initial part in $G_h$. More precisely, we need to know which type of component (in Lemma \ref{lem:formofcc}) every vertex is in, and if a vertex is in a subgraph of type 1, that is, the subgraph contains a cycle, we need to know the length of that cycle, and the length of the path leading to the cycle from the vertex.

\paragraph{Acknowledgements.} The authors are grateful to Joseph Almog for his help and discussions on the topic. Most grateful we are to the referees of the preliminary version of this paper who helped us realize (together with Almog) what we actually had proved.

\end{document}